\documentclass[journal,draftcls,onecolumn,12pt,twoside]{IEEEtranTCOM}
\usepackage{amsmath}
\usepackage{amssymb}
\usepackage{amsfonts}
\usepackage{mdframed}
\usepackage{amscd}
\usepackage{mathrsfs}
\usepackage{graphics}
\usepackage{graphicx}
\usepackage{color}
\usepackage{url}
\usepackage{bm}
\usepackage{algorithm}
\usepackage{algorithmic}

\newtheorem{theorem}{Theorem}

\newtheorem{definition}{Definition}
\newtheorem{proposition}{Proposition}

\newtheorem{example}{Example}

\begin{document}

\title{Practical Encoders and Decoders for Euclidean Codes from Barnes-Wall Lattices}
\author{J. Harshan*, Emanuele Viterbo*\thanks{*The authors are with the Department of Electrical and Computer Systems Engineering, Monash University, Melbourne, Australia-3168. Email:harshan.jagadeesh@monash.edu, emanuele.viterbo@monash.edu. $^{\dagger}$The author is with the Department of Communications and Electronics, Telecom ParisTech, Paris, France, Email:belfiore@enst.fr.
Parts of this work are in the proceedings of IEEE International Symposium on Information Theory (ISIT) 2012 held at Cambridge, MA, USA, and International Symposium on Mathematical Theory of Networks and Systems (MTNS) 2012 held at Melbourne, Australia.}, and Jean-Claude Belfiore$^{\dagger}$}

\maketitle
 
\begin{abstract}
In this paper, we address the design of high spectral-efficiency Barnes-Wall (BW) lattice codes which are amenable to low-complexity decoding in additive white Gaussian noise (AWGN) channels. We propose a new method of constructing complex BW lattice codes from linear codes over polynomial rings, and show that the proposed construction provides an explicit method of bit-labelling complex BW lattice codes. To decode the code, we adapt the low-complexity sequential BW lattice decoder (SBWD) recently proposed by Micciancio and Nicolosi. First, we study the error performance of SBWD in decoding the infinite lattice, wherein we analyze the noise statistics in the algorithm, and propose a new upper bound on its error performance. We show that the SBWD is powerful in making correct decisions well beyond the packing radius. Subsequently, we use the SBWD to decode lattice codes through a novel noise-trimming technique. This is the first work that showcases the error performance of SBWD in decoding BW lattice codes of large block lengths.
\end{abstract}
 



\section{Introduction}
\label{sec1}

Ever since random coding schemes were demonstrated to approach the capacity of additive white Gaussian noise (AWGN) channels \cite{CoT}, enormous research has taken place to find \emph{structured} coding schemes which can accomplish the same job. The need for structured coding schemes is to facilitate simpler analysis of the code structure and to achieve reduced complexity in encoding and decoding. 
A well known method of obtaining structured codes is to carve out a finite set of lattice points from dense lattices \cite{SFS}-\cite{FoU}. Such codes are referred to as lattice codes, and are usually obtained as a set of coset representatives of a suitable quotient lattice. Further, the lattice codes have the advantage of inheriting most of the code properties from the parent lattice, and as a result, the choice of the lattice is crucial to the performance of the code.

\subsection{Motivation and contributions}

\indent In this paper, we are interested in carving lattice codes from Barnes-Wall (BW) lattices \cite{BaW}, \cite{WJP}. Our goal is to construct efficient BW lattice codes of large block lengths which work with low-complexity encoders and decoders. 
In particular, efficient low-complexity decoders for complex BW lattices are readily available in \cite{MiN}, \cite{GrP}. Therefore, if lattice codes from complex BW lattices are employed for communication over AWGN channels, then the decoders of \cite{MiN}, \cite{GrP} can be used to recover information with low computational complexity.

In \cite{MiN}, the authors have proposed two low-complexity implementations of the bounded distance decoder for BW lattices, namely ({\em i}) the sequential bounded distance decoder, and ({\em ii})~the parallel bounded distance decoder. Inspired by the parallel decoder in \cite{MiN}, list decoders for BW lattices have been recently proposed in \cite{GrP}. We note that the parallel decoders of \cite{MiN} and \cite{GrP} have low-complexity only when implemented on sufficiently large number of parallel processors. If the above decoders are implemented on a single processor, then the complexity advantages are lost, and specifically, the complexity of the list decoder grows larger than that of the sequential decoder in \cite{MiN}. Since we are interested in lattice codes of large block lengths, we focus on the sequential bounded distance decoder which seems more suitable for implementation (see Section \ref{sec5_subsec2} for more details on the complexity advantages of sequential bounded distance decoder over the list decoder). The sequential decoder in \cite{MiN} was proven to correct any error up to the packing radius. However, the possibility of correct decoding is not known when the received vector falls outside the bounded decoding ball of packing radius. In a nutshell, the exact error performance of the decoder is not known. The existence of this low-complexity decoder has motivated us to study its error performance, and use it to decode BW lattice codes. We refer to this decoder as the sequential BW lattice decoder (SBWD). The contribution of this paper on the construction and decoding of complex BW lattices are given below.

\begin{enumerate}
\item We introduce Construction $A^{\prime}$ of lattices which enables us to generate some well structured lattices from linear codes over \emph{polynomial rings} \cite{HVB}. As an immediate application, we apply Construction $A^{\prime}$ to obtain BW lattices of dimension $2^{m}$ for any $m \geq 1$. The proposed method is yet another construction of BW lattices (shown in Section \ref{sec3}) and shows a new connection between codes over polynomial rings and lattices. We show that the proposed construction provides an explicit method of obtaining and bit-labelling complex BW lattice codes. 
\item We study the error performance of the SBWD in AWGN channels. Since the error performance of the SBWD depends on the error performance of the underlying soft-input Reed-Muller (RM) decoders, we study the error performance of the soft-input RM decoder as used in the SBWD. First, we use the Jacobi-Theta functions \cite{theta_func} to characterize the virtual binary channels that arise in the decoding process. Subsequently, we study the noise statistics in the algorithm, and provide an upper bound on the error performance of the soft-input RM decoders. Through computer simulations, we obtain the error performance of the SBWD, and show that the decoder is powerful in making correct decisions well beyond the packing radius \cite{HVB1} (see Table I in Section \ref{sec5} for the effective radius of the SBWD decoder). This is the first work that showcases the error performance of SBWD in decoding BW lattices of large block lengths. 
\item To decode the lattice code in AWGN channels, we employ the SBWD along with a noise trimming technique, wherein the components of the received vector are appropriately scaled before passing it to the SBWD. With the noise-trimming technique, the SBWD is forced to decode to a codeword in the code which in turn improves the error performance. We refer to this decoder as the \emph{BW lattice code decoder} (BWCD). We obtain the bit error rate (BER) of the BWCD for codes in complex dimension $4, 16$, and $64$, and show that the BWCD outperforms the SBWD by $0.5$ dB.
\end{enumerate}

\subsection{Prior work on Barnes-Wall lattices}
\label{sec1_subsec1}

The BW lattices \cite{BaW} is a special family of $N$-dimensional lattices that exist when $N$ is a power of $2$. These lattices were originally discovered as a solution to finding extreme quadratic forms in 1959 \cite{BaW}. Only in 1983, the now well known connection between BW lattices and Reed-Muller codes was discovered by \cite{SlB}. This connection is found in several works \cite{CoS}, \cite{GDF}, \cite{SaA} in different forms. Other than its construction through Reed Muller codes, the generator matrices of the BW lattices are also known to be obtained though Kronecker products \cite{NRS}, \cite{MiN}. 

\indent In 1989, G.D. Forney proposed a low-complexity bounded distance decoding algorithm for Leech lattices \cite{Forney_bd}. As a generalization, in the same paper, a similar algorithm has been shown to be applicable in decoding all Construction $D$ lattices. As BW lattices are known to be obtained through Construction $D$ \cite{CoS}, bounded distance decoders for BW lattices were known in principle since \cite{Forney_bd}. Only in the 1990's, explicit bounded distance decoders for BW lattices were implemented for dimension up to $32$, and numerical results on the error performance were reported \cite{MoJ}, \cite{Yucel}. In 2008, Micciancio and Nicolosi \cite{MiN} have proposed two low-complexity implementations of bounded distance decoder for BW lattices, namely ({\em i}) the sequential bounded distance decoder, and ({\em ii}) the parallel bounded distance decoder. If $N = 2^{m}$ denotes the dimension of complex BW lattice, the worst-case complexity of the decoders has been shown to be $O\left(N\mbox{log}^{2}(N)\right)$ and $O\left(\mbox{log}^{2}(N)\right)$ for the fully sequential decoder and the fully parallel decoder, respectively. For the fully sequential decoder, the algorithm is assumed to be implemented on a single processor, whereas for the fully parallel decoder, the algorithm is assumed to be implemented on $N^{2}$ parallel processors. Inspired by the fully parallel implementation in \cite{MiN}, list decoders for BW lattices have been recently proposed in \cite{GrP} where the list decoder outputs a list of BW lattice points within any given radius from the target vector. The complexity of the list decoder is shown to be polynomial in the dimension of the lattice, and polynomial in the list size, which is a function of the Euclidean radius. Note that the SBWD exploits the Construction $D$ structure of BW latices as a multilevel code of nested binary Reed-Muller (RM) codes, and decodes each RM code through a successive interference cancellation technique. On the other hand, the list decoder does not exploit construction $D$ structure of BW lattices, and hence, does not need the support of any soft-input RM decoders. 

The rest of this paper is organized as follows: In Section \ref{sec2}, we provide a background on lattice constructions from linear codes. In Section \ref{sec3}, we introduce Construction $A^{\prime}$ of complex BW lattices. In Section \ref{sec4}, we study the error performance of the SBWD,  while in Section \ref{sec5} and Section \ref{sec6}, we use the SBWD to decode the BW lattice code. Finally, in Section \ref{sec7}, we conclude this paper and provide some directions for future work. 

\textit{\textbf{Notations}:} Throughout the paper, boldface letters and capital boldface letters are used to represent vectors and matrices, respectively. For a complex matrix $\textbf{X}$, the matrices $\textbf{X}^T$, $\Re(\textbf{X})$ and $\Im (\textbf{X})$ denote the transpose, real part and imaginary part of $\textbf{X}$, respectively. The set of integers, real numbers, and complex numbers are denoted by ${\mathbb Z}$, $\mathbb{R}$, and ${\mathbb C}$, respectively. We use $i$ to represent $\sqrt{-1}.$ For an $n$-length vector $\textbf{x}$, we use $x_{j}$ to represent the $j$-th component of $\textbf{x}$. Cardinality of a set $\mathcal{S}$ is denoted by $|\mathcal{S}|$. Magnitude of a complex number $x$ is denoted by $|x|$. The number of ways of picking $n$ objects out of $m$ objects is denoted by $C^{m}_{n}$. The symbol $\lceil\cdot\rfloor$ denotes the nearest integer of a real number, and we set $\lceil a + 0.5 \rfloor = a$ for any $a \in \mathbb{Z}$. Finally, we use $\mbox{Pr}(\cdot)$ to denote the probability operator.

\section{Background on Lattice Construction using Linear Codes}
\label{sec2}

A complex lattice $\Lambda$ over $\mathbb{Z}[i]$ is a discrete subgroup of $\mathbb{C}^{n}$ \cite{CoS}. Alternatively, $\Lambda$ is a $\mathbb{Z}[i]$-module generated by a basis set $\{ \textbf{v}_{1}, \textbf{v}_{2}, \ldots, \textbf{v}_{n} ~|~ \textbf{v}_{j} \in \mathbb{C}^{n} \}$ as $\Lambda = \left\lbrace \sum_{j = 1}^{n} q_{j}\textbf{v}_{j} ~|~ \forall q_{j} \in \mathbb{Z}[i] \right\rbrace.$ 
It is well known that dense lattices can be obtained via binary linear codes \cite{CoS}. Depending on the structure of the underlying linear codes, lattice construction can be categorized into different types. In this section, we recall two well known constructions for the case of complex lattices~\cite{CoS}.

\begin{flushleft}
\textbf{Construction $A$:}\\
\end{flushleft}

\begin{definition}
\label{def0_1}
A complex lattice $\Lambda$ is obtained by Construction $A$ from the binary linear code $\mathcal{C}$ if $\Lambda$ can be represented as
\begin{equation}
\label{cons_A}
\Lambda= (1+i) \mathbb{Z}[i]^{n} \oplus \mathcal{L}_{0},
\end{equation}
where $\mathcal{L}_{0} = \{ \psi (\textbf{c}) ~|~ \forall \textbf{c} \in \mathcal{C} \} \subseteq \Lambda$ is a lattice code obtained by the component-wise mapping $\psi:\mathbb{F}_{2} \rightarrow \mathbb{Z}[i]$ given by $\psi (0) = 0$ and $\psi (1) = 1$ on the alphabet of $\mathcal{C}$, where $\mathbb{F}_{2} = \{0, 1\}$.  
\end{definition}

\begin{flushleft}
\textbf{Construction $D$:}\\
\end{flushleft}

\begin{definition}
\label{def_cons_d}
A complex lattice $\Lambda$ is obtained by Construction $D$ from a family of nested binary linear codes $\mathcal{C}_{m-1} \supseteq \mathcal{C}_{m-2} \supseteq \ldots \supseteq \mathcal{C}_{1} \supseteq \mathcal{C}_{0}$ if $\Lambda$ can be represented as
\begin{equation}
\label{cons_D}
\Lambda= (1+i)^{m} \mathbb{Z}[i]^{n} \oplus (1+i)^{m-1}\mathcal{L}_{m-1} \oplus \cdots \oplus (1+i)\mathcal{L}_{1} \oplus \mathcal{L}_{0},
\end{equation}
where $\mathcal{L}_{j} = \{ \psi (\textbf{c}) ~|~ \forall \textbf{c} \in \mathcal{C}_{j} \}$ is obtained by the component-wise mapping $\psi:\mathbb{F}_{2} \rightarrow \mathbb{Z}[i]$ given by $\psi (0) = 0$ and $\psi (1) = 1$ on the alphabet of $\mathcal{C}_{j}$.  
\end{definition}

\indent A BW lattice can be obtained via construction $D$ as a $\mathbb{Z}[i]$ lattice as follows \cite{GDF}. Suppose we want to construct the complex lattice $BW_{2^{m}}$ of dimension $2^{m}$ where $m \geq 1$, let $\mathcal{RM}(r, m)$ be the binary Reed-Muller (RM) code (Sec. 3.7, Ch. 3, \cite{Blahut}) of length $2^m$ and of order $0 \leq r \leq m$. Then, $BW_{2^{m}}$ can be constructed as
\begin{equation}
\label{construction_d}
BW_{2^{m}} = \left \{ (1+i)^{m} \textbf{a} + \sum_{r=0}^{m-1} (1+i)^{r}\psi(\textbf{c}_{r}) ~|~ \forall \textbf{c}_{r} \in \mathcal{RM}(r, m), \forall \textbf{a} \in \mathbb{Z}[i]^{2^{m}} \right\}
\end{equation}
where $\psi(\cdot)$ is as given in Definition \ref{def_cons_d}. For notational convenience, we also write \eqref{construction_d} as
\begin{equation}
\label{construction_d_easy}
BW_{2^{m}} = (1+i)^{m}\mathbb{Z}[i]^{2^{m}} \oplus \bigoplus_{r=0}^{m-1} (1+i)^{r}\mathcal{RM}(r, m).
\end{equation}
This method generates $BW_{2^{m}}$ as a multi-level structure of nested RM codes and hence it falls under Construction $D$ \cite{CoS}. 

\begin{flushleft}
\textbf{Generalized construction $A$ \cite{AmF} \cite{FoV}:}\\
\end{flushleft}

\indent Apart from Construction $D$, the BW lattice codes can be obtained by the generalized construction $A$. For the complex BW lattice $BW_{2^{m}}$, let $\textbf{G}_{\small{\mbox{BW}}} \in \mathbb{C}^{N \times N}$ denote the generator matrix in the triangular form, where the rows $\{ \textbf{g}_{1}, ~\textbf{g}_{2}, \ldots, \textbf{g}_{N} \}$ of $\textbf{G}_{\small{\mbox{BW}}}$ forms a basis set of $BW_{2^{m}}$, where $N = 2^{m}$. Let $d_1, d_{2}, \ldots ,d_N$ represent the diagonal elements, where $d_{j} = (1+i)^{m_{j}}$ for some integer $m_{j} \geq 0$, and $d = (1+i)^{\max m_{j}}$. For this lattice construction, one can easily map binary data to lattice points as follows:
\begin{enumerate}
\item \textbf{Bit Labelling:} Map $\mbox{log}_{2}(L_{j})$ information bits to $a_{j} \in \mathbb{Z}[i] /p_{j}\mathbb{Z}[i]$ where $p_{j} = \frac{d}{d_j}$ and $L_{j}$ is the cardinality of $\mathbb{Z}[i] /p_{j}\mathbb{Z}[i]$.
\item \textbf{Encoding:} Using $\{ a_{1}, a_{2}, \ldots, a_{N} \}$, a lattice point is obtained as $\sum_{j = 1}^{N} a_{j} \textbf{g}_{j}$.
\item \textbf{Shaping:} Since $\Lambda = d \mathbb{Z}[i]^{N} + \mathcal{L}$, a lattice point within $\mathcal{L}$ can be obtained as $\bar{\textbf{x}} = \textbf{x}$ mod $d \mathbb{Z}[i]^{N}$. 
\end{enumerate}

\begin{flushleft}
\textbf{Motivation for Construction $A^{\prime}$:}\\
\end{flushleft}

\indent In the bit-labelling step above, binary digits have to be mapped to the symbols of $\mathbb{Z}[i] /p_{i}\mathbb{Z}[i]$. Some of the well-known bit-labelling methods include gray-mapping and set-partitioning based methods.\footnote{Unlike uncoded communication, gray-mapping on $\mathbb{Z}[i] /p_{i}\mathbb{Z}[i]$ is not necessarily optimal since it does not guarantee that the neighbouring lattice points in the lattice code are separated by maximum number of information bits. Efficient bit labelling of lattice codes is a separate problem of its own and is out of the scope of this work.} Unlike the case of real integer lattice, $\mathbb{Z}[i] /p_{i}\mathbb{Z}[i]$ is an arbitrary subset of $\mathbb{Z}[i]$, and bit mapping to $\mathbb{Z}[i] /p_{j}\mathbb{Z}[i]$ is not straightforward unless the set of representatives for $\mathbb{Z}[i] /p_{j}\mathbb{Z}[i]$ is chosen with good shaping property. Through Construction $A^{\prime}$, we facilitate bit-labelling on complex integers by using the truncated binary expansion of the elements of $\mathbb{Z}[i] /p_{j}\mathbb{Z}[i]$ over the base $1+i$ \cite{NiK}. With this, the bits labelled on $a_{j} \in \mathbb{Z}[i] /p_{j}\mathbb{Z}[i]$ are nothing but the bits in the truncated binary expansion of $a_{j}$. To assist the bit-labelling step, we use polynomial rings over $\mathbb{F}_{2}$ to represent the elements of $\mathbb{Z}[i] /p_{j}\mathbb{Z}[i]$. For the encoding step, we use a linear code over polynomial rings, and obtain the lattice points as embedding of the codewords a of linear code into the Euclidean space. Finally, for the shaping step, we propose an appropriate mapping on $\mathbb{Z}[i]$ which provides a label code with appropriate shaping property, i.e., we explicitly provide a method of bit-labelling complex BW lattices.
Our construction is an extension of Construction $A$ and hence we refer to it as Construction $A^{\prime}$. We now define polynomial rings and codes over polynomial rings. 

\begin{definition}
(Ch. 4 in \cite{Blahut})
We define the polynomial quotient ring $\mathcal{U}_{m} = \mathbb{F}_{2}[u] \diagup u^{m}$ in variable $u$ for any $m \geq 1$ as 
\begin{equation*}
\mathcal{U}_{m} = \left \lbrace \sum_{k = 0}^{m-1} b_{k} u^{k} \mbox{ mod } u^{m} ~|~ b_{k} \in \mathbb{F}_{2} \right \rbrace,
\end{equation*}
with regular polynomial addition and multiplication over $\mathbb{F}_{2}$ coefficients along with the quotient operation $u^{m} = 0,$ which is equivalent to cancelling all the terms of degree greater than or equal to $m$.
\end{definition}

\begin{definition}
A linear code $\mathcal{C}$ over $\mathcal{U}_{m}$ is a subset of $\mathcal{U}_{m}^{n}$ which can be obtained through a generator matrix $\textbf{G} \in \mathcal{U}_{m}^{k \times n}$ as
\begin{equation*}
\mathcal{C} = \{ \textbf{z} \textbf{G} ~|~ \forall \textbf{z} \in \mathcal{U}_{m}^{k} \},
\end{equation*}
for some $k \leq n$ and the matrix multiplication is over the ring $\mathcal{U}_{m}$.
\end{definition}

\section{Construction $A^{\prime}$ of BW Lattice}
\label{sec3}

We now introduce Construction $A^\prime$ in the following definition.

\begin{definition}
\label{def1}
A complex lattice $\Lambda$ is obtained by Construction $A^{\prime}$ from a linear code $\mathcal{C}$ over $\mathcal{U}_{m}$ for some $m \geq 1$ if $\Lambda$ can be written as
\begin{equation}
\label{cons_A}
\Lambda = \Phi(u^{m}) \mathbb{Z}[i]^{n} + \mathcal{EC},
\end{equation}
where $\mathcal{EC} = \{ \Phi (\textbf{c}) ~|~ \forall \textbf{c} \in \mathcal{C} \} \subseteq \mathbb{Z}[i]^{n}$ is a lattice code obtained from the linear code $\mathcal{C}$ through the mapping $\Phi:\mathcal{U}_{m} \rightarrow \mathbb{Z}[i]$ given by
\begin{equation*}
\Phi\left(\sum_{j = 0}^{m-1} b_{j} u^{j}\right) = \sum_{j = 0}^{m-1} \psi(b_{j}) \left(\Phi(u)\right)^{j},
\end{equation*}
such that $\psi:\mathbb{F}_{2} \rightarrow \mathbb{Z}[i]$ given by $\psi (0) = 0$ and $\psi (1) = 1$, and $\Phi(u) = 1+i$.
\end{definition}

\indent Note that Construction $A$ can be obtained as a special case from Construction $A^{\prime}$ when $m = 1$, wherein the embedding operation $\Phi$ coincides with $\psi$ given in Definition \ref{def_cons_d}. In the following subsections, we use Construction $A^{\prime}$ to obtain complex BW lattices of dimension $2^{m}$ for any $m \geq 1$ by embedding a linear code $\mathcal{C}$ (denoted by $\mathcal{C}_{2^{m}}$) over the quotient ring $\mathcal{U}_{m}$ to a lattice code $\mathcal{EC}$ (denoted by $\mathcal{EC}_{2^{m}}$).

\subsection{Linear codes for construction $A^{\prime}$}
\label{sec3_subsec1}

\indent In order to obtain $BW_{2^{m}}$ as Construction $A^{\prime}$, we first need to find a suitable linear code $\mathcal{C}_{2^{m}}$ over the ring $\mathcal{U}_{m}$. We propose such a linear code which can be obtained by the following the generator matrix
%
%
%
\begin{equation*}
\label{gen_matrix}
\textbf{G}_{2^{m}} = \left[\begin{array}{cc}
1 & 1\\
0 & u\\
\end{array}\right]^{\otimes m},
\end{equation*}
where the tensor operation is over the ring $\mathcal{U}_{m}$. 
\begin{example}
To obtain $BW_{4}$, the linear code $\mathcal{C}_{4}$ can be generated using
\begin{equation*}
\textbf{G}_{4} = \left[\begin{array}{cccc}
1 & 1 & 1 & 1\\
0 & u & 0 & u\\
0 & 0 & u & u\\
0 & 0 & 0 & 0\\
\end{array}\right] \in \mathcal{U}_{2}^{4 \times 4}.
\end{equation*}
\end{example}
$~~$\\

\noindent \textbf{Encoding of linear code $\mathcal{C}_{2^{m}}$}\\
\indent By using $\textbf{G}_{2^m}$ as a matrix over $\mathcal{U}_{m}$, the code $\mathcal{C}_{2^{m}}$ is obtained as follows: Let $\textbf{z} \in \mathcal{U}_{m}^{2^{m}}$, i.e., the $j$-th component of $\textbf{z}$ is given by 
\begin{equation}
\label{inpt_bit_poly}
z_{j} = \sum_{k = 0}^{m-1} b_{k,j} u^{k},
\end{equation} 
where $b_{k,j} \in \mathbb{F}_{2}$ for all $k, j$. Using $\textbf{z}$ and $\textbf{G}_{2^m}$, the code $\mathcal{C}_{2^{m}} \subseteq \mathcal{U}_{m}^{2^{m}}$ can be obtained as 
\begin{equation}
\label{construction_a}
\mathcal{C}_{2^{m}} = \left\{ \textbf{x} = \textbf{z} \textbf{G}_{2^m} ~|~ \forall \textbf{z} \in \mathcal{U}_{m}^{2^{m}} \right\},
\end{equation}
where the matrix multiplication is over $\mathcal{U}_{m}$. 

\indent We now provide an example for the proposed encoding technique, showing the positions of the information bits that get encoded to the codewords of $\mathcal{C}_{2^{m}}$.

\begin{example}
For $m = 2$, the input vector $\textbf{z}$ and the generator matrix $\textbf{G}_{4}$ are of the form,
\begin{equation*}
\textbf{z}^{T} = \left[\begin{array}{c}
b_{0,1} + b_{1,1}u \\
b_{0,2}\\
b_{0,3}\\
0\\
\end{array}\right] \mbox{ and } \textbf{G}_{4} = \left[\begin{array}{cccc}
1 & 1 & 1 & 1\\
0 & u & 0 & u\\
0 & 0 & u & u\\
0 & 0 & 0 & 0\\
\end{array}\right].
\end{equation*}
\end{example}
$~~$\\

\indent We define the rate of the linear code $\mathcal{C}_{2^{m}}$ as the ratio of the number of information bits per codeword and the length of the code (which is also known as the spectral-efficiency of the code).  

\begin{proposition}
\label{prop1}
The rate of the code $\mathcal{C}_{2^{m}}$ is $\frac{m}{2}$.
\end{proposition}
\begin{proof}
Each component of $\textbf{z}$ carries $m$ information bits in the variables $b_{k, j}$ as shown in \eqref{inpt_bit_poly}. This amounts to a total of $m2^{m}$ bits carried by $\textbf{z}$. However, since the matrix multiplication is over $\mathcal{U}_{m}$, not all the information bits $b_{k, j}$ are encoded as codewords of $\mathcal{C}_{2^{m}}$ (since $u^{k} = 0$ for $k \geq m$). Using the structure of $\textbf{G}_{2^m}$ it is possible to identify the indices $(k, j)$ of information bits $b_{k, j}$ which get encoded into the codewords of $\mathcal{C}_{2^{m}}$ as follows. Let the set $\mathcal{I}_{q}$ denote the indices of the rows of $\textbf{G}_{2^m}$ whose components take values $0$ or $u^{q}$ for each $q = 0, 1, \ldots, m-1$. Due to the quotient operation $u^{m} = 0$, the components of $\textbf{z}$ which are in the index set $\mathcal{I}_{q}$ are restricted to be of the form, $$z_{j} = \sum_{k = 0}^{m-1-q} b_{k,j}u^{k} ~\forall j \in \mathcal{I}_{q}.$$ For example, $z_{1} = \sum_{k = 0}^{m-1} b_{k,1}u^{k}$ and $z_{2^{m}} = 0$. Using the structure of $\textbf{G}_{2^m}$ we observe that the cardinality of $\mathcal{I}_{q}$ denoted by $|\mathcal{I}_{q}|$ is $C^{m}_{q}$, and hence we find the total number of information bits per codeword of $\mathcal{C}_{2^{m}}$ as $\sum_{k = 0}^{m-1} (m - k) C^{m}_{k} = \frac{m}{2}2^{m}.$
\end{proof}

We now show the equivalence of our encoding technique to Construction $D$. In other words, the following theorem shows that the codewords generated in \eqref{construction_a} can be uniquely represented as vectors of a multi-level code of nested RM codes.
\begin{theorem}
The codewords generated in \eqref{construction_a} can be uniquely represented as codewords obtained through Construction $D$.
\end{theorem}
\begin{proof}
See the proof of Theorem 1 in \cite{HVB}.
\end{proof}

Till now, we have presented the linear code $\mathcal{C}_{2^{m}}$ and its encoding technique over the quotient ring $\mathcal{U}_{m}$. Now, we discuss the embedding operation of $\mathcal{C}_{2^{m}}$ into the Euclidean space. By using the map $\Phi(u) = 1+i$ on $\mathcal{C}_{2^{m}}$, we get the lattice code $\mathcal{EC}_{2^{m}}$.
Note that $\mathcal{EC}_{2^{m}}$ can be used as a \emph{tile} in constructing the BW lattice, i.e., $BW_{2^{m}}$ can be obtained by replicating $\mathcal{EC}_{2^{m}}$ in $\mathbb{Z}[i]^{2^{m}}$ as $BW_{2^{m}} = (1+i)^{m}\mathbb{Z}[i]^{2^{m}} + \mathcal{EC}_{2^{m}}.$ It can be verified that $\mathcal{EC}_{2^{m}}$ is an arbitrary subset of $BW_{2^{m}}$ and does not have cubic shaping. In Fig. \ref{tiling}, we plot the complex points generated as $\left\{ \sum_{r=0}^{m-1} (1+i)^{r}b_{r} ~|~ b_{r} \in \{0, 1\} \right\}$ for $m = 10$. Note that the points generated by $\sum_{r=0}^{m-1} (1+i)^{r}b_{r}$ are marked in black, whereas the points in other shades correspond to the shifted version of $\sum_{r=0}^{m-1} (1+i)^{r}b_{r}$ by constants $(1+i)^{m}, i(1+i)^{m}$ and $(1+i)(1+i)^{m}$.

\begin{figure}
\centering
\includegraphics[width=3.5in]{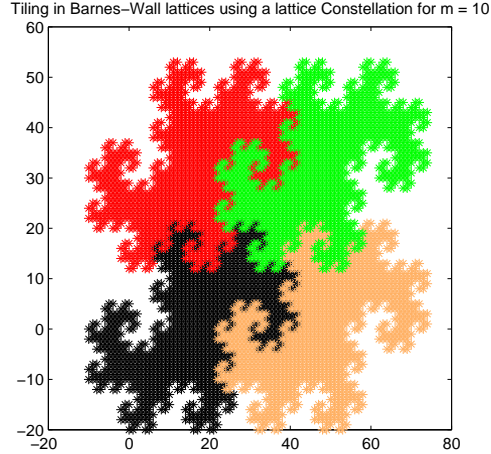}
\caption{Filling the complex plane using the tile generated by $\sum_{r=0}^{m-1} (1+i)^{r}b_{r}$ for $m = 10$.}
\label{tiling}
\end{figure}

Note that the code $\mathcal{EC}_{2^{m}}$ does not have good shaping, we observe that the average transmit power of the scheme is not small. To fix this problem, we propose a one-to-one mapping $\phi$ on $\mathcal{EC}_{2^{m}}$ to obtain a new lattice code denoted by $\mathcal{L}_{2^{m}}$ such that it has good shaping property.

\subsection{BW lattice codes with cubic shaping}
\label{sec3_subsec3}

Here, we propose a one-to-one mapping $\phi$ on $\mathcal{EC}_{2^{m}}$ to obtain a new lattice code $\mathcal{L}_{2^{m}}$ which has the cubic shaping property when $m$ is even, and the rectangular shaping property when $m$ is odd. For any $\textbf{x} = \left[ x_{1}, x_{2}, x_{3}, \ldots, x_{2^{m}} \right] \in \mathcal{EC}_{2^{m}}$, the mapping $\phi$ operates on each component of $\textbf{x}$ as,
\begin{eqnarray}
\label{cubic_shaping_map}
\phi(x_{j}) = \left\{ \begin{array}{ccccc}
x_{j} \mbox{ mod } 2^{\frac{m}{2}}, \mbox{ when } m \mbox{ is even};\\
\varphi\left(x_{j} \mbox{ mod } 2^{\frac{m+1}{2}} \right), \mbox{ when } m \mbox{ is odd},\\
\end{array} 
\right.
\end{eqnarray}
where $\varphi(\cdot)$ is defined on $\mathbb{Z}_{2^{\frac{m+1}{2}}}[i]$ as
{\small
\begin{eqnarray}
\label{new_map_odd_m}
\varphi(z) = \left\{ \begin{array}{ccccc}
z, \mbox{ when } \Im(z) < 2^{\frac{m-1}{2}};\\
z + \left( 2^{\frac{m-1}{2}} - i2^{\frac{m-1}{2}}\right), \mbox{ when } \Re(z) < 2^{\frac{m-1}{2}} \\ ~~~~~~~~~~~\mbox{ and } \Im(z) \geq 2^{\frac{m-1}{2}};\\
z - \left( 2^{\frac{m-1}{2}} + i2^{\frac{m-1}{2}}\right), \mbox{ when } \Re(z) \geq 2^{\frac{m-1}{2}}\\ ~~~~~~~~~~~\mbox{ and } \Im(z) \geq 2^{\frac{m-1}{2}}.\\
\end{array} 
\right.
\end{eqnarray}
}
The mapping $\phi$ guarantees the following property on $\mathcal{L}_{2^{m}}$:
\begin{eqnarray}
\label{mod_box}
\mathcal{L}_{2^{m}} \subseteq \left\{ \begin{array}{ccccc}
\left\{\mathbb{Z}_{2^{\frac{m}{2}}}[i]\right\}^{2^{m}}, \mbox{ if } m \mbox{ is even};\\~~\\
\left\{\mathbb{Z}_{2^{\frac{m+1}{2}}}\right\}^{2^{m}} + i\left\{\mathbb{Z}_{2^{\frac{m-1}{2}}}\right\}^{2^{m}}, \mbox{ if } m \mbox{ is odd}.\\
\end{array} 
\right.
\end{eqnarray}
From \eqref{mod_box}, note that each component of the vector in $\mathcal{L}_{2^{m}}$ is in a cubic box and a rectangular box, when $m$ is even and odd, respectively. In Fig. \ref{cubic_shaping}, we present the complex points $\sum_{r=0}^{m-1} (1+i)^{r}b_{r}$ with and without the mapping $\phi$ for $m = 10$. With this, the lattice code $\mathcal{L}_{2^{m}}$ can be obtained from $\mathcal{C}_{2^{m}}$ through the composition map 
\begin{equation}
\label{map_chi}
\chi = \phi(\Phi(\cdot)),
\end{equation}
where $\Phi$ and $\phi$ are given in Definition \ref{def1} and \eqref{cubic_shaping_map} respectively. The following proposition shows that $\chi(\cdot)$ is a one-to-one map on $\mathcal{C}_{2^{m}}$

\begin{figure}
\centering
\includegraphics[width=3.5in]{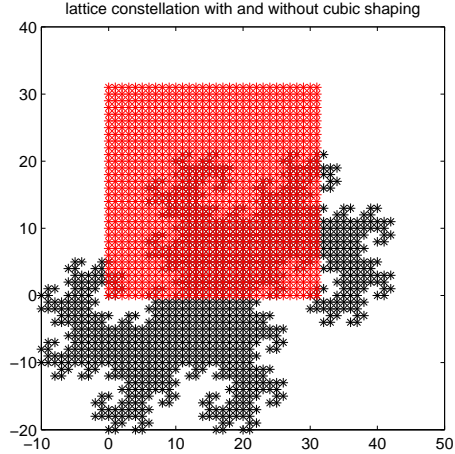}
\caption{Complex points generated by $\sum_{r=0}^{m-1} (1+i)^{r}b_{r}$ and $\phi(\sum_{r=0}^{m-1} (1+i)^{r}b_{r})$ for $m = 10$.}
\label{cubic_shaping}
\end{figure}

\begin{proposition}
The mapping $\chi$ given in \eqref{map_chi} is one-to-one.
\end{proposition}
\begin{proof}
Since $\chi$ is a composition mapping of $\Phi$ and $\phi$, and $\Phi(\cdot)$ is a substitution operation using binary representation of complex numbers over the base $(1+i)$, we have to prove that $\phi$ given in \eqref{cubic_shaping_map} is one-to-one. Here, we provide the proof when $m$ is even. For any $\textbf{x}_{1}, \textbf{x}_{2} \in \mathcal{EC}_{2^{m}}$ such that $\textbf{x}_{1} \neq \textbf{x}_{2}$, we prove that $\phi(\textbf{x}_{1}) \neq \phi(\textbf{x}_{2})$. Applying the modulo operation in \eqref{cubic_shaping_map}, $\textbf{x}_{j}$ satisfies $\textbf{x}_{j} = 2^{\frac{m}{2}} \textbf{r}_{j} + \phi(\textbf{x}_{j})$ for each $j = 1, 2$, where $\phi(\textbf{x}_{j}) \in \mathcal{L}_{2^{m}}$ and $\textbf{r}_{j} \in \mathbb{Z}[i]^{2^m}$. This implies 
\begin{eqnarray}
\label{phi_of_vector}
\phi(\textbf{x}_{j}) = \textbf{x}_{j} - 2^{\frac{m}{2}} \textbf{r}_{j} = \textbf{x}_{j} + (1+i)^m \textbf{r}_{j}',
\end{eqnarray}
for some $\textbf{r}_{j}' \in \mathbb{Z}[i]^{2^m}$. The second equality follows as $(1+i)^m = a2^{\frac{m}{2}}, \mbox{ where } a \in \{ 1, -1, i, -i \}$. Further, since each component of $\textbf{x}_{j}$ is of the form $\sum_{r=0}^{m-1} (1+i)^{r}b_{r}$ for $b_{r} \in \{0, 1\}$, the R.H.S of \eqref{phi_of_vector} is nothing but the binary decomposition of $\phi(\textbf{x}_{j})$ over the base $(1+i)$. Since the radix representation over $(1+i)$ is unique, we have $\phi(\textbf{x}_{1}) = \phi(\textbf{x}_{2})$ only if $\textbf{x}_{1} = \textbf{x}_{2}$. This completes the proof when $m$ is even. The one-to-one nature of $\phi$ can be proved on the similar lines when $m$ is odd.
\end{proof}

The above proposition implies that mapping $\phi$ provides a new lattice code with better shaping property. The following theorem shows that $\mathcal{L}_{2^{m}}$ can be used as a tile to obtain BW lattices.
\begin{theorem}
\label{thm2}
The lattice code $\mathcal{L}_{2^{m}}$ and the lattice $BW_{2^{m}}$ are related as $BW_{2^{m}} = (1+i)^m\mathbb{Z}[i]^{2^m} \oplus \mathcal{L}_{2^{m}}.$
\end{theorem}
\begin{proof}
See the proof of Theorem 2 in \cite{HVB}.
\end{proof}

Using the results of Theorem \ref{thm2}, $BW_{2^{m}}$ is given by $BW_{2^{m}} = (1+i)^m \mathbb{Z}[i]^{2^{m}} \oplus \mathcal{L}_{2^{m}},$ where $\mathcal{L}_{2^{m}}$ is the lattice code obtained from $\mathcal{C}_{2^{m}}$ through the mapping $\chi = \phi(\Phi(\cdot))$ on $\mathcal{U}_{m}$.

\section{On the Error Performance of the SBWD}
\label{sec4}

In this section, we study the error performance of the SBWD in decoding the infinite BW lattice. In \cite{MiN}, it is shown that for $\textbf{x} \in BW_{2^{m}}$, if there exists $\textbf{y} \in \mathbb{C}^{2^{m}}$ such that $d^{2}_{min}(\textbf{x}, \textbf{y}) \leq \frac{N}{4}$, where $N = 2^{m}$, then the SBWD correctly finds (or decodes) the lattice point $\hat{\textbf{x}} = \textbf{x}$. In the context of using SBWD in AWGN channels, the vector $\textbf{y}$ corresponds to $\textbf{y} = \textbf{x} + \textbf{n}$, where $\textbf{x} \in BW_{2^{m}}$ and $n_{j} \sim \mathcal{CN}(0, \sigma^{2}) ~\forall j$. This implies that the codeword error rate (CER) of the SBWD given by $\mbox{Pr}(\hat{\textbf{x}} \neq \textbf{x})$ is upper bounded as
\begin{equation*}
\label{sub}
\mbox{Pr}(\hat{\textbf{x}} \neq \textbf{x}) \leq  \mbox{Pr}\left(|\textbf{n}|^{2} > \frac{N}{4}\right).
\end{equation*} 
Note that $\frac{\sqrt{N}}{2}$ is the packing radius of $BW_{2^{m}}$, and hence the above bound is the well known \emph{sphere upper bound} (SUB) \cite{ViB}. In \cite{MiN}, the focus was only on the complexity of the decoder but not on the analysis of the tightness of the SUB. In other words, the possibility of correct decision is not known when $|\textbf{n}|^{2} > \frac{N}{4}$. We study the error performance and show that the decoder is powerful in making correct decisions well beyond the packing radius. Without loss of generality, we study the error performance when the zero lattice point is transmitted. We analyze the SBWD algorithm and point out the reason for the improvement in the error performance (with reference to the SUB). We first recall the SBWD algorithm of \cite{MiN}.

\begin{flushleft}
\textbf{The Sequential BW Lattice Decoding Algorithm:}
\end{flushleft}
\vspace{0.2cm}
\begin{mdframed}
\noindent \textbf{function} SEQBW$(r$, $\textbf{y})$\\
$~~~$\textbf{if} $\textbf{y} \in \mathbb{C}^{N}$ and $N \leq 2^{r}$\\
$~~~$ $~~~$ return $\lceil \textbf{y} \rfloor$;\\
$~~~$\textbf{else}\\
$~~~$ $~~~$ $\textbf{b} = \lceil \Re(\textbf{y}) \rfloor + \lceil \Im(\textbf{y}) \rfloor \mbox{ mod } 2$;\\
$~~~$ $~~~$ $\rho = 1 - 2(\mbox{max}\left(|\lceil \Re(\textbf{y}) \rfloor - \Re(\textbf{y})|, |\lceil \Im(\textbf{y}) \rfloor| - \Im(\textbf{y})\right))$;\\
$~~~$ $~~~$ $\hat{\textbf{c}} = \mbox{RMDEC}(r, \textbf{b}, \rho)$;\\
$~~~$ $~~~$ $\textbf{v} = \mbox{SEQBW}(r+1, (\textbf{y} - \hat{\textbf{c}})/(1+i))$;\\
$~~~$ $~~~$ return $\hat{\textbf{\textbf{c}}} + (1+i)\textbf{v}$;\\
$~~~$\textbf{end if}\\
\textbf{end function}
\end{mdframed}
\vspace{0.3cm}
\indent The above decoder is a successive interference cancellation (SIC) type decoder which exploits the BW lattice structure as a multi-level code of nested RM codes (as per Construction $D$). At each level, the algorithm uses a variant of the soft-input RM decoder \cite{ScB} (denoted by the function RMDEC which is given as Algorithm 3 in \cite{MiN}) to decode, and cancel the RM codeword at that level. Therefore, the error performance of the SBWD is fundamentally determined by the error performance of the underlying soft-input RM decoders. In particular, we have
\begin{equation}
\label{lattice_rm_error_relation}
\mbox{Pr}(\hat{\textbf{x}} \neq \textbf{x}) = \mbox{Pr}\left(\bigcup_{r}  \mathcal{E}(\hat{\textbf{c}}_{r} \neq \textbf{c}_{r})\right),
\end{equation}
where $\mathcal{E}(\hat{\textbf{c}_{r}} \neq \textbf{c}_{r})$ denotes an error event while decoding $\mathcal{RM}(r, m)$. Hence, it is important to compute $\mbox{Pr}(\hat{\textbf{c}}_{r} \neq \textbf{c}_{r})$ for each $\mathcal{RM}(r, m)$. Along that direction, it is necessary to model the effective binary channel induced for each RM code $\mathcal{RM}(r, m)$. We propose a model for such a binary channel which is accurate for $r = 0$, while for $r \neq 0$, it neglects the error propagation in the SIC decoder algorithm. To decode the RM code at each level, a hard-decision binary value $b_{j}$ is obtained from $y_{j}$ as
\begin{equation}
\label{mod_1_plus_i}
b_{j} = \lceil \Re(y_{j}) \rfloor + \lceil \Im(y_{j}) \rfloor \mbox{ mod } 2.
\end{equation}
Due to the combination of the round and the modulo operation (henceforth referred to as the round-modulo operation) in \eqref{mod_1_plus_i}, the codewords of $\mathcal{RM}(r, m)$ are passed through a virtual binary channel with the cross-over probability given by,
\begin{equation*}
P_{c} = \mbox{Pr}(b_{j} = 1 ~|~ c_{j} = 0),
\end{equation*}
where $\textbf{c} \in \mathcal{RM}(r, m)$. Since the zero lattice point is transmitted, $\textbf{c}$ is the all zero codeword for each $\mathcal{RM}(r, m)$, and hence the relevant cross-over probability is $\mbox{Pr}(b_{j} = 1 ~|~ c_{j} = 0)$. The following theorem shows that $P_{c}$ can be upper bounded by a Jacobi-Theta function \cite{theta_func}.

\begin{figure}[h]
\centering
\includegraphics[width=3in]{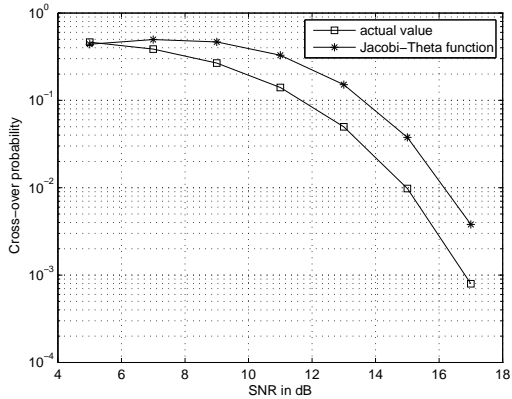}
\caption{Comparison of the cross-over probability with the upper bound using the Jacobi-Theta function}
\label{PC_JACOBI}
\end{figure}

\begin{theorem} The cross-over probability $P_{c}$ induced by the round-modulo operation in \eqref{mod_1_plus_i} is upper bounded as 
\begin{equation}
\label{jacobi_theta_bound}
P_{c} \leq \left( e^{-\frac{1}{4\sigma^{2}}}\right) \vartheta \left(\frac{i 4}{\pi \sigma^{2}}, \frac{i }{\pi \sigma^{2}}\right),
\end{equation}
where $\vartheta \left(z, \tau \right)$ is the Jacobi-Theta function given by
\begin{equation*}
\vartheta \left(z, \tau \right) = \sum_{a = -\infty}^{\infty} e^{\pi i a^{2} \tau + 2 \pi i a z}.
\end{equation*}
\end{theorem}
\begin{proof}
We first compute $P_{c}$, and then propose an upper bound. To assist compute $P_{c}$, we compute the probability that $\Re(y_{j})$ (or $\Im(y_{j})$) falls within an interval $(z - 0.5, z +0.5]$ centred around an integer $z$, when $c_{j} = 0$. Since the additive noise is circularly symmetric, it is sufficient to calculate the above probability for either $\Re(y_{j})$ or $\Im(y_{j})$. We use $y$ to denote either $\Re(y_{j}$) or $\Im(y_{j})$. For the odd integer case, we have
\begin{eqnarray}
\label{delta_odd}
P_{\mbox{\small{o}}} & \triangleq & \sum_{a = -\infty}^{\infty} \mbox{Pr}\left( 2a + 0.5 < y \leq 2a + 1.5 \right), \nonumber \\
&  = & \sum_{a = -\infty}^{\infty} \left[ \int_{2a + 0.5}^{2a + 1.5}  \mbox{P}_{\textbf{y}}(y) dy \right], \nonumber \\
& = & \sum_{a = -\infty}^{\infty} \left[ Q\left(\frac{2a + 0.5}{\sigma/\sqrt{2}}\right) - Q\left(\frac{2a + 1.5}{\sigma/\sqrt{2}}\right) \right],
\end{eqnarray}
where $\mbox{P}_{\textbf{y}}(y)$ is the probability density function of $y$, $Q(x) = \frac{1}{\sqrt{2 \pi}} \int_{x}^{\infty} e^{-\frac{u^2}{2}} du$,  and $\sigma^2/2$ is the variance of $y$.
For the even integer case, we have
\begin{eqnarray}
\label{delta_even}
P_{\mbox{\small{e}}} & \triangleq & \sum_{a = -\infty}^{\infty} \mbox{Pr}\left( 2a - 0.5 < y \leq 2a + 0.5 \right), \nonumber \\
&  = & \sum_{a = -\infty}^{\infty} \left[ \int_{2a - 0.5}^{2a + 0.5}  \mbox{P}_{\textbf{y}}(y) dy \right], \nonumber \\
&  = & \sum_{a = -\infty}^{\infty} \left[ Q\left(\frac{2a - 0.5}{\sigma/\sqrt{2}}\right) - Q\left(\frac{2a + 0.5}{\sigma/\sqrt{2}}\right) \right].
\end{eqnarray} 
Note that $b_{j}$ is $1$ whenever $\lceil \Re(y_{j}) \rfloor + \lceil \Im(y_{j}) \rfloor$ is an odd number. This can happen when ({\em i}) $\lceil \Re(y_{j}) \rfloor$ is odd and $\lceil \Im(y_{j}) \rfloor$ is even, or
({\em ii}) $\lceil \Re(y_{j}) \rfloor$ is even and $\lceil \Im(y_{j}) \rfloor$ is odd. From \eqref{delta_odd} and \eqref{delta_even}, we can write
\begin{eqnarray}
P_{c} & = & P_{\mbox{\small{o}}}(1 - P_{\mbox{\small{o}}}) + (1 - P_{\mbox{\small{o}}})P_{\mbox{\small{o}}},\\
& = & 2P_{\mbox{\small{o}}} - 2(P_{\mbox{\small{o}}})^{2}.
\end{eqnarray}
By dropping the term $2(P_{\mbox{\small{o}}})^{2}$, we upper bound $P_{c}$ as
\begin{eqnarray}
P_{c} & \leq & 2P_{\mbox{\small{o}}}, \nonumber \\
& \leq & 2 \sum_{a = -\infty}^{\infty} \left[ Q\left(\frac{2a + 0.5}{\sigma/\sqrt{2}}\right)\right] \label{drop_q},\\
& \leq & \sum_{a = -\infty}^{\infty} e^{-\frac{(2a + 0.5)^{2}}{\sigma^{2}}} \label{apply_chernoff}, \\
& = & e^{-\frac{(0.5)^{2}}{\sigma^{2}}} \sum_{a = -\infty}^{\infty} e^{\frac{-4a^{2} - 2a}{\sigma^{2}}}, \nonumber\\
& = & \left( e^{-\frac{1}{4\sigma^{2}}}\right) \vartheta \left(\frac{i 4}{\pi \sigma^{2}}, \frac{i }{\pi \sigma^{2}}\right), \nonumber
\end{eqnarray}
where the bound in \eqref{drop_q} comes from dropping the terms of the form $Q\left(\frac{2a + 1.5}{\sigma/\sqrt{2}}\right)$ in \eqref{delta_odd}, and the bound in \eqref{apply_chernoff} is due to the Chernoff bound $Q(x) \leq \frac{1}{2} e^{\frac{-x^{2}}{2}}$. 
\end{proof}

Note that the Jacobi-Theta function can be evaluated at any pair ($\tau$, $z$). In Fig. \ref{PC_JACOBI}, the empirical values of $P_{c}$ are presented along with the bound in \eqref{jacobi_theta_bound} for various values of $\mbox{SNR} = \frac{1}{\sigma^{2}}$. We point out that the bound is not tight due to the Chernoff-bound on each $Q(\cdot)$ function.

\indent It is well known that $P_{c}$ determines the error-performance of a hard decision decoder. Since we have a soft-input decoder, we need to obtain the relevant statistics on the soft inputs. 
We now study the soft-input RM decoder used in the SBWD. Unlike the codewords of RM code in \cite{ScB}, the RM codewords at each level of BW lattice take values over $\{0, 1\}$.
The soft-input used for the RM decoder is $\rho = 1 - 2\textbf{d},$ where $\textbf{d} = \mbox{max}\left(|\lceil \Re(\textbf{y}) \rfloor - \Re(\textbf{y})|, |\lceil \Im(\textbf{y}) \rfloor - \Im(\textbf{y})|\right)$. Also, unlike the soft metric in \cite{ScB}, $\rho_{j}$ is bounded in the interval $[0, 1]$. This is because $d_{j} \in [0, 0.5]$, which is a result of the round-modulo operation in \eqref{mod_1_plus_i}. One could imagine $\textbf{b}$ and $\rho$ to be obtained from the received vector in a virtual additive noise channel, wherein each component of the received vector is always within a distance of $0.5$ from either $0$ or $1$. Therefore, if $\textbf{c}$ denotes a RM codeword at a particular level of the transmitted BW lattice point, then the effective noise $\textbf{n}^{eff}$ as seen by the soft-input RM decoder at that level is of the form,
\begin{eqnarray}
\label{effective_noise}
n^{eff}_{j} = \left\{ \begin{array}{ccccc}
d_{j}, \mbox{ when } b_{j} = c_{j};\\
1 - d_{j}, \mbox{ when } b_{j} \neq c_{j};\\
\end{array} 
\right.
\end{eqnarray}
for $1 \leq j \leq N$. Note that $n^{eff}_{j}$ has bounded support in the interval $\left[0, 1\right]$. 
%
%
For an analogy with respect to the model in \cite{ScB}, the code alphabet $\{ 0, 1\}$ in \cite{MiN} corresponds to the code alphabet $\{ -1, 1\}$ in \cite{ScB} and the effective noise $\textbf{n}^{eff}$ in \cite{MiN} corresponds to the AWGN in \cite{ScB}. At each level of the BW lattice, the lattice code $(1+i)^{r}\mathcal{RM}(r, m)$ for any $0 \leq r \leq m - 1$ has the minimum squared Euclidean distance of $N$. By using the proposition in Section IV.A of \cite{ScB}, the probability of incorrect decision of the soft-input RM decoder at each level of SBWD is upper bounded as shown in the proposition below.

\begin{proposition}
The codeword error rate $\mbox{Pr}(\hat{\textbf{c}}_{r} \neq \textbf{c}_{r})$ for each $\mathcal{RM}(r, m)$ is upper bounded as,
\begin{equation}
\label{new_upper_bound}
\mbox{Pr}(\hat{\textbf{c}}_{r} \neq \textbf{c}_{r}) \leq \mbox{Pr}\left(|\textbf{n}^{eff}|^{2} > \frac{N}{4}\right) \mbox{ for } r = 0, 1, \ldots, m-1.\\
\end{equation}
\end{proposition}

\indent It is important to note that the above bound is different from $\mbox{Pr}(|\textbf{n}|^{2} > \frac{N}{4})$ since $\textbf{n}$ is Gaussian distributed. We do not have closed form expression on the distribution of either $n_{j}^{eff}$ or $|n_{j}^{eff}|^{2}$. In Fig. \ref{hist}, we display the histogram of the realizations of $n_{j}^{eff}$ for various values of $\sigma^{2}$, when the zero RM codeword is the transmitted. Note that for $\sigma^{2}$ = $0$ dB, the histogram of $n_{j}^{eff}$ has the triangular shape centred around $0.5$, which implies a very high (close to 0.5) cross-over probability when obtaining the hard decision vector $\textbf{b}$. On the other hand, at lower values of $\sigma^{2}$, the distribution is skewed towards zero indicating smaller cross-over probability.

\begin{figure}
\centering
\includegraphics[width=3.5in]{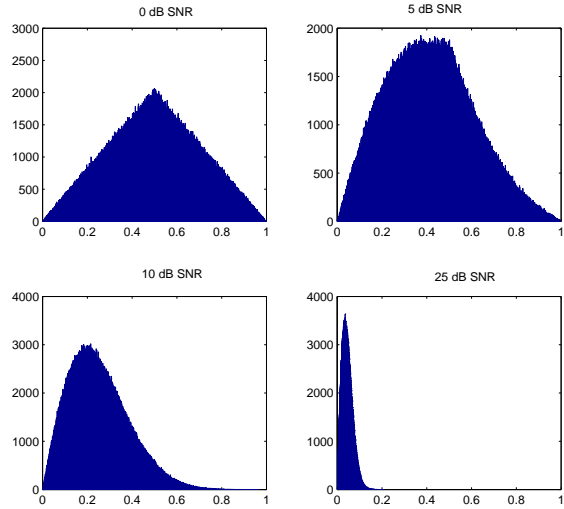}
\caption{Histogram of $n^{eff}_{j}$ when SNR = $\frac{1}{\sigma^{2}}$ takes the values 0 dB, 5 dB, 10 dB and 25 dB.}
\label{hist}
\end{figure}

%
%
%
%

\section{SBWD to Decode BW Lattice Code $\mathcal{L}_{2^{m}}$ for AWGN Channel}
\label{sec5}

In this section, we discuss the use of SBWD to decode the lattice code $\mathcal{L}_{2^{m}}$. First, we describe a method to transmit the codewords of $\mathcal{L}_{2^{m}}$. For any $\textbf{x} \in \mathcal{L}_{2^{m}}$, the transmitted vector is of the form\footnote{The transmitted vector is offset by a constant $c$ towards the origin to reduce the average transmit energy.} 
\begin{equation}
\label{off_set}
\textbf{x}_{t} = \left( 2\textbf{x} - c \right),
\end{equation}
where
\begin{eqnarray}
c = \left\{ \begin{array}{ccccc}
\left(2^{\frac{m}{2}} - 1\right) + i \left(2^{\frac{m}{2}} - 1\right), \mbox{ when } m \mbox{ is even};\\
\left(2^{\frac{m+1}{2}} - 1\right)+ i \left(2^{\frac{m-1}{2}} - 1\right), \mbox{ when } m \mbox{ is odd}.\\
\end{array} 
\right.
\end{eqnarray}
Using the scale and the shift operation in \eqref{off_set}, each component of $\textbf{x}_{t}$ takes value from the regular $2^{m}$-QAM constellation. In particular, the QAM constellation is square and non-square when $m$ is even and odd, respectively. When $\textbf{x}_{t}$ is transmitted, the received vector $\bar{\textbf{y}}$ is given by
\begin{equation}
\label{actual_awgn_channel}
\bar{\textbf{y}} = \textbf{x}_{t} + \bar{\textbf{n}},
\end{equation}
where $\bar{\textbf{n}}$ is the AWGN with $\bar{n}_{j} \sim \mathcal{CN}(0, \sigma^{2}) ~\forall j$. In this section, SNR of the channel is defined as $E_{s}/\sigma^2$, where $E_{s}$ denotes the average energy of $2^{m}$-QAM constellation. With the inverse operation to \eqref{off_set} as
$\textbf{y} = \frac{1}{2}\bar{\textbf{y}} + c$, the equivalent AWGN channel becomes
\begin{equation}
\label{eq_awgn_channel}
\textbf{y} = \textbf{x} + \textbf{n},
\end{equation}
where $\textbf{x} \in \mathcal{L}_{2^{m}}$ and $n_{j} \sim \mathcal{CN}(0, \frac{\sigma^{2}}{4})$. We use the SBWD \cite{MiN} on \eqref{eq_awgn_channel} to decode the lattice code $\mathcal{L}_{2^{m}}$.  When a codeword of $\mathcal{L}_{2^{m}}$ is transmitted, the SBWD decodes to a lattice point in the infinite lattice $BW_{2^{m}}$. In such a decoding method, irrespective of whether the decoded lattice point falls in the code or not, the information bits can be recovered from the decoded RM codewords at every level of SBWD (as shown in the algorithm in Sec. \ref{sec4}).

\subsection{Simulation results on the codeword error rate (CER) of SBWD}
\label{sec5_subsec1}

In this subsection, we present the CER of the SBWD along with some upper bounds and lower bounds. For the simulation results, we use $\mbox{SNR} = E_{s}/\sigma^2$, where $E_{s}$ denotes the average energy of the regular $2^{m}$-QAM constellation. In each of Fig. \ref{ILD_m_2}-\ref{ILD_m_10}, we present ({\em i}) the CER of the SBWD, ({\em ii}) the SUB (Section IV.D, \cite{ViB}), ({\em iii}) the sphere lower bound (SLB) (Section IV.D, \cite{ViB}), ({\em iv}) the CER in decoding $\mathcal{RM}(0, m)$ at the first level of the SBWD, and ({\em v}) the upper bound on the CER in decoding $\mathcal{RM}(0, m)$ given by $\mbox{Pr}(|\textbf{n}^{eff}|^{2} > \frac{N}{4})$ (obtained through simulation results by empirically generating $\textbf{n}^{eff}$).

\begin{figure}[h]
\centering
\includegraphics[width=3.5in]{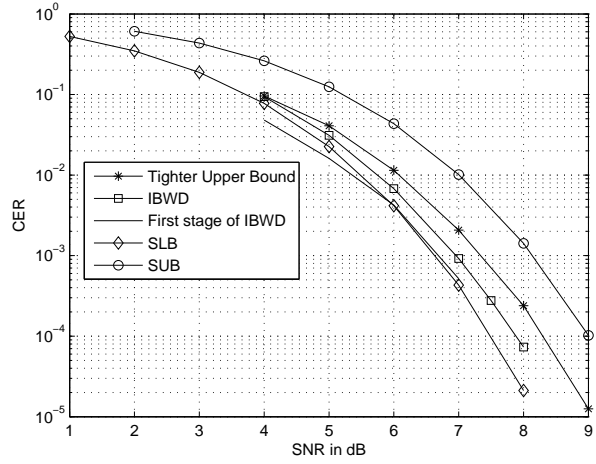}
\caption{CER of SBWD decoding $BW_{4}$.}
\label{ILD_m_2}
\end{figure}


\begin{figure}[h]
\centering
\includegraphics[width=3.5in]{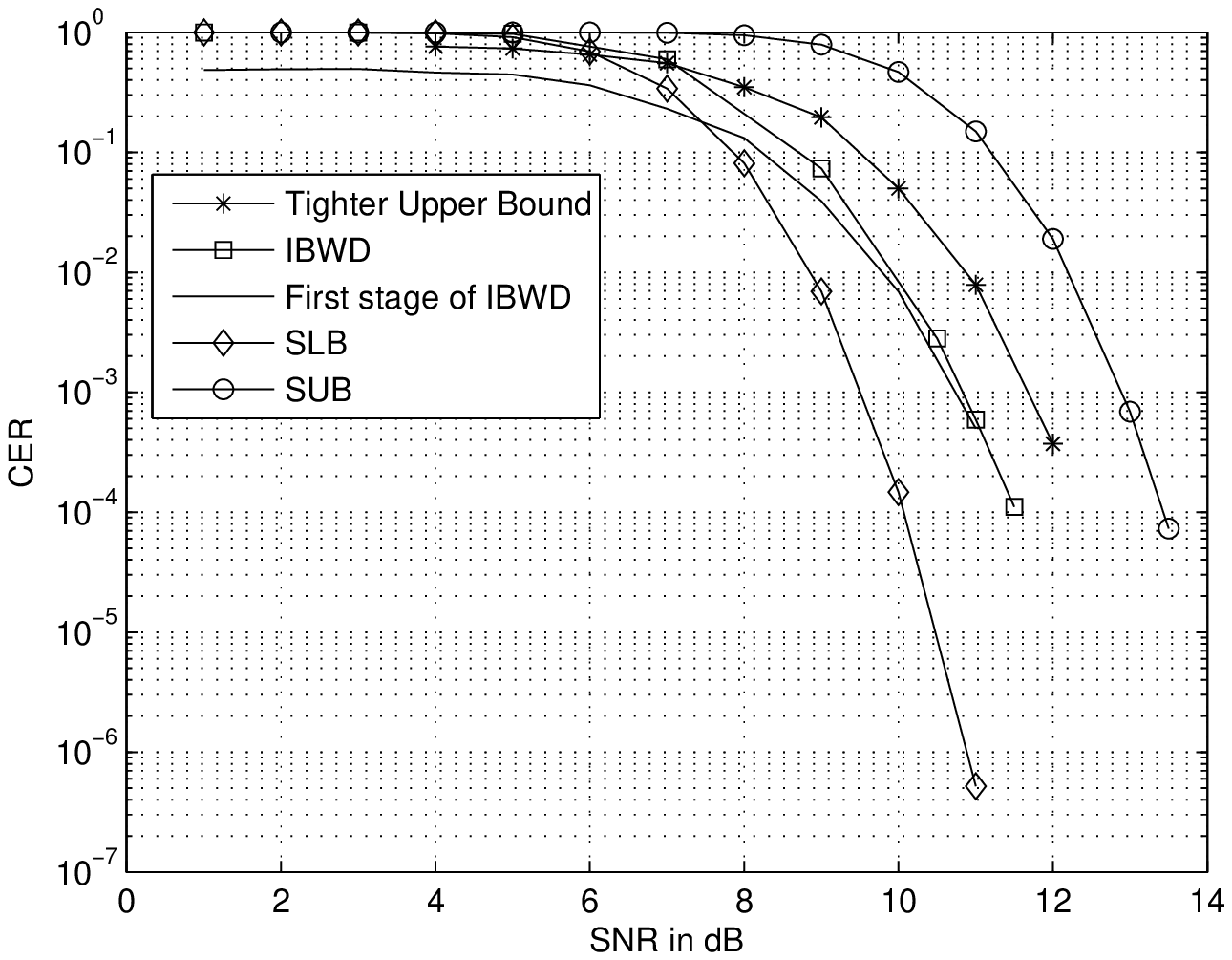}
\caption{CER of SBWD for decoding $BW_{16}$.}
\label{ILD_m_4}
\end{figure}


\begin{figure}[h]
\centering
\includegraphics[width=3.5in]{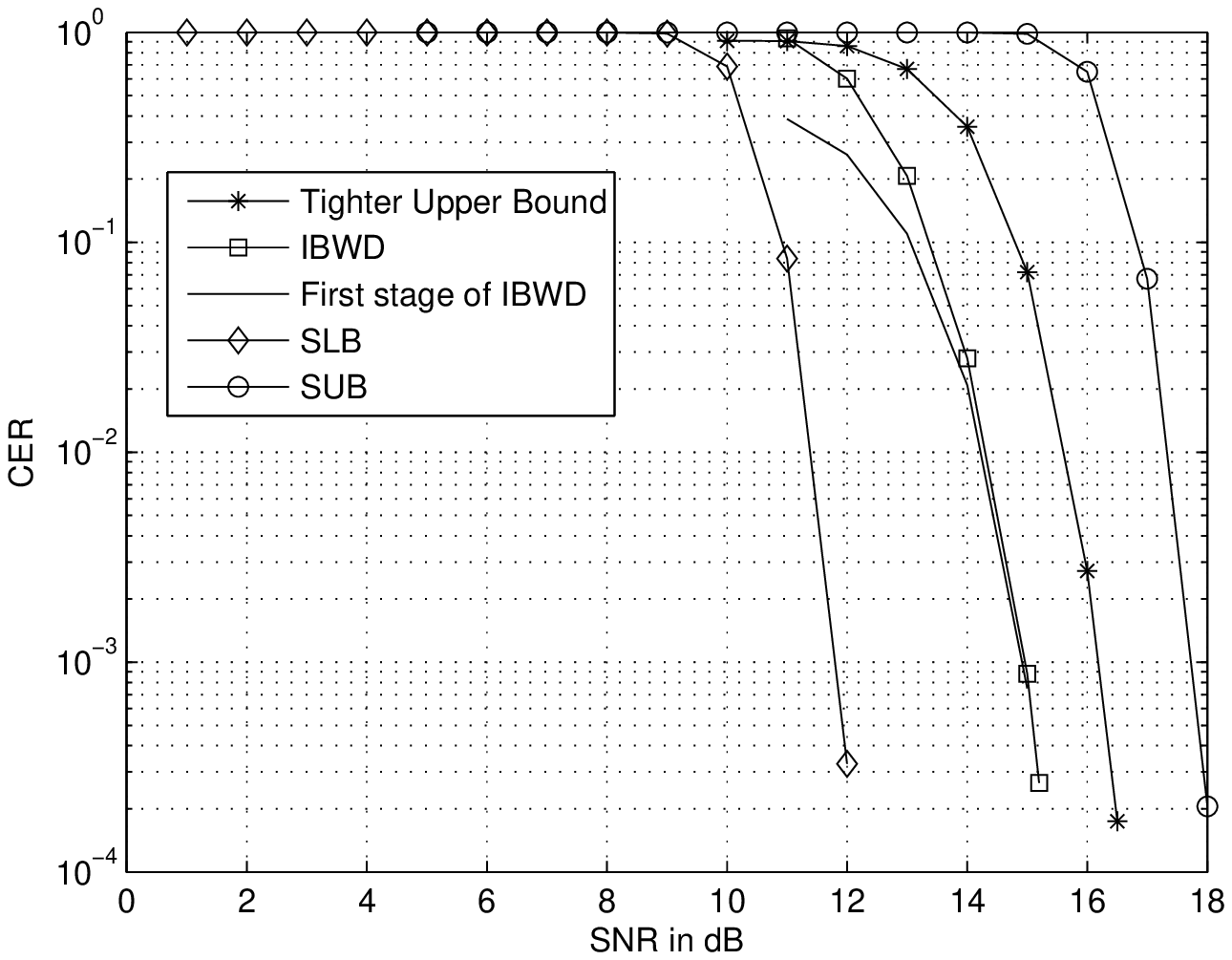}
\caption{CER of SBWD for decoding $BW_{64}$.}
\label{ILD_m_6}
\end{figure}


\begin{figure}[h]
\centering
\includegraphics[width=3.5in]{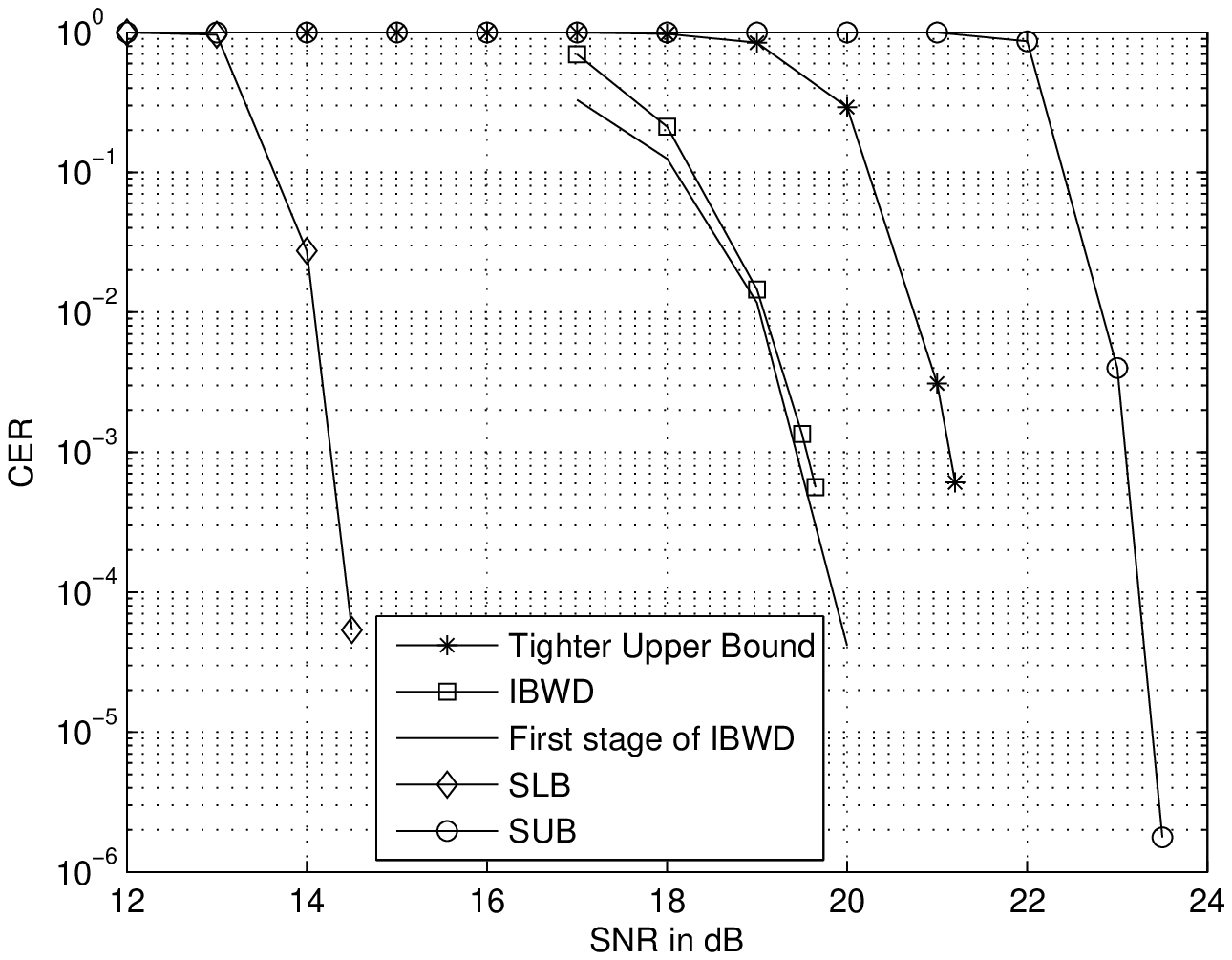}
\caption{CER of SBWD for decoding $BW_{256}$.}
\label{ILD_m_8}
\end{figure}


\begin{figure}[h]
\centering
\includegraphics[width=3.5in]{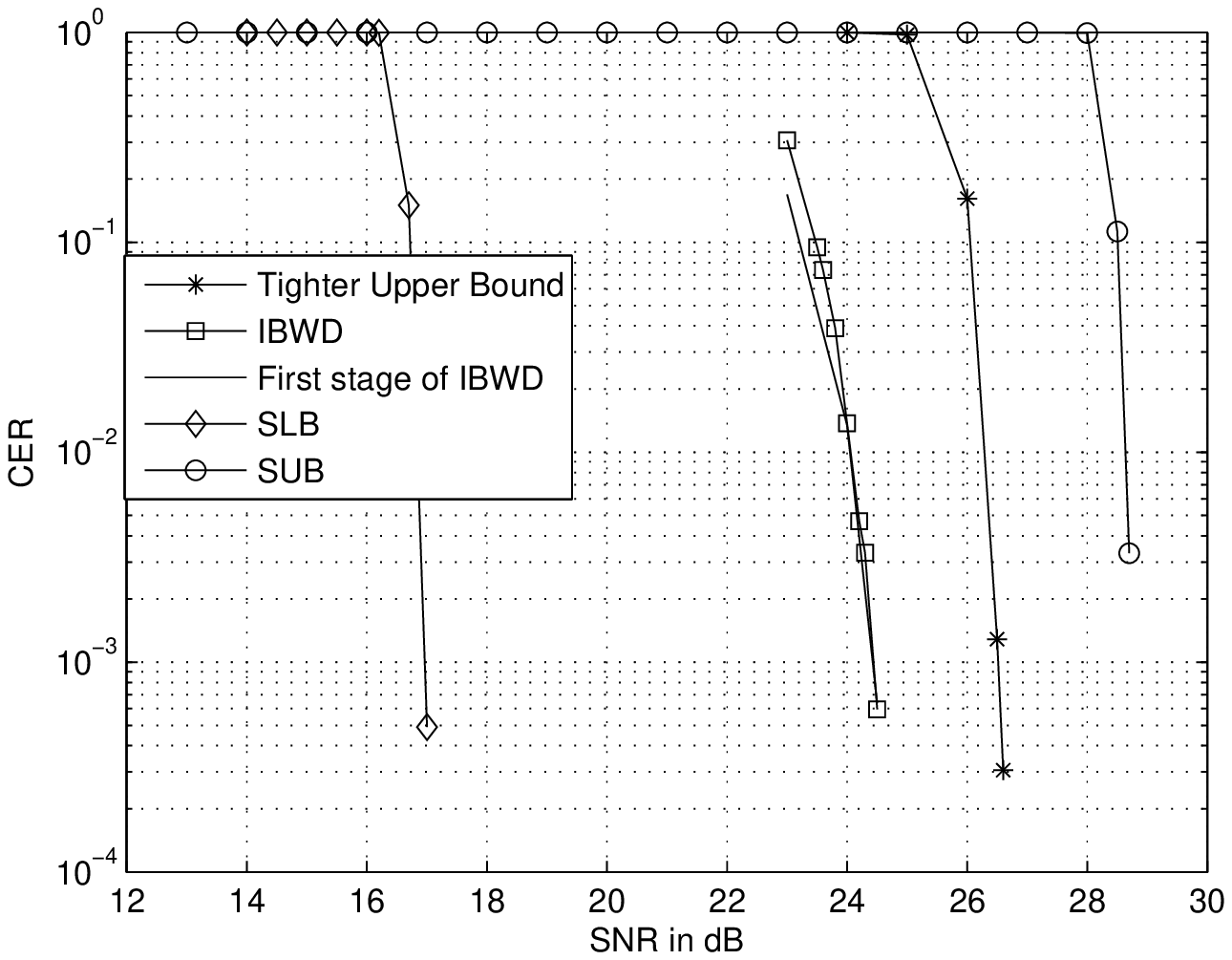}
\caption{CER of SBWD for decoding $BW_{1024}$.}
\label{ILD_m_10}
\end{figure}



From Fig. \ref{ILD_m_2}-\ref{ILD_m_10}, we make the following observations: the SUB is not a tight upper bound on the CER of SBWD. Also, $\mbox{Pr}(|\textbf{n}^{eff}|^{2} > \frac{N}{4})$ is an upper bound on the CER of SBWD, and in particular, it is a tighter upper bound than the SUB. The CER of the soft-input RM decoder for $\mathcal{RM}(0, m)$ is a tight lower bound on the CER of the SBWD. This implies that if there is no error at the first level of the decoder, then with high probability, there will be no errors at subsequent levels of the soft-input RM decoder. In summary, the simulation results highlight that the SBWD is quite powerful in making correct decisions even beyond the packing radius, and the deviation from the SUB increases for larger dimensions. As a result SBWD can be employed to efficiently decode lattice codes of large block lengths with low-complexity. This behaviour in the error performance of SBWD was not known in the literature.

\subsection{Comparing the complexity of the SBWD with the list decoder in \cite{GrP}}
\label{sec5_subsec2}

In this subsection, we compare the complexity of the SBWD with the BW list decoder \cite{GrP}. For a fair comparison, we assume that the list decoder is implemented on a single processor. On a single processor, the complexity of the SBWD is $O(N\mbox{log}^{2}(N))$, whereas the complexity of the list decoder is $O(N^{2})(\mathit{l}(m, \eta))^2$, where $\mathit{l}(m, \eta)$ is the worst case list size at a relative squared distance of $\eta$ (the relative squared distance is the squared Euclidean distance normalized by the dimension of the lattice). We compare the complexity of the two decoders for a codeword error rate of $10^{-3}$. In particular, we first approximate the error performance of the SBWD as a bounded distance decoder for some radius $\bar{\eta}$, and then compute the complexity of the list decoder with the corresponding value of $\bar{\eta}$. In Table \ref{table_complexity}, we display the lower bound (as given in Theorem 1.3 in \cite{GrP}) on the complexity of the list decoder to achieve the error performance of SBWD. The table shows that the list decoder has higher complexity than the SBWD to achieve the same performance. In summary, for single processor implementation, SBWD can be preferred to the list decoder to decode BW lattice codes of large block lengths. However, for codeword error rates lower than that of SBWD, the list decoder has to be used, preferably on parallel processors. Table \ref{table_complexity} also shows the potential of SBWD to decode well beyond the relative squared distance of $\eta = 0.25$. For complex dimensions of $256$ and $1024$, the effective radius of SBWD is as high as $\frac{N}{2}$ and $\frac{2N}{3}$, respectively.

\begin{center}
\begin{table}
\caption{Complexity of the list decoder \cite{GrP} to achieve the performance of SBWD}
\begin{center}
\begin{tabular}{|c|c|c|c|c|c|c|c|c|c|c|}
\hline Dimension $N$  & $\bar{\eta}$ & A lower bound on $N^{2}(\mathit{l}(m, \bar{\eta}))^2$ & $N\mbox{log}^{2}(N)$ (complexity of SBWD)\\
\hline 4 &  0.33 & 16 & 16 \\
\hline 16 &  0.4 & 256 & 256 \\
\hline 64 &  0.48 & 4096 &  2304 \\
\hline 256 &  0.56 & 262144 & 16384 \\
\hline 1024 & 0.67 & $1.07 \times 10^{9}$ & 102400\\
\hline
\end{tabular}
\end{center}
\label{table_complexity}
\end{table}
\end{center}

\section{Noise Trimming Technique for the SBWD}
\label{sec6}

When a codeword of $\mathcal{L}_{2^{m}}$ is transmitted, the SBWD decodes to a lattice point in the infinite lattice $BW_{2^{m}}$. In such a decoding method, irrespective of whether the decoded lattice point falls in the code or not, the information bits can be recovered from the decoded RM codewords at every level of SBWD (as shown in the algorithm in Sec. \ref{sec4}). To further improve the error performance, we force the SBWD to specifically decode to a codeword in the lattice code, and subsequently recover the information bits, with more reliability. We refer to such a decoder as the BW lattice code decoder (BWCD). We use a technique that forces the SBWD to decode to a codeword in the lattice code $\mathcal{L}_{2^{m}}$. We refer to this technique as the noise trimming technique, which exploits the structure of $\mathcal{L}_{2^{m}}$. From \eqref{mod_box}, we know that each component of a codeword is within a rectangular box $\mathcal{B} \subseteq \mathbb{C}$. In particular, the box $\mathcal{B}$ shares its edges with $\mathbb{Z}_{2^{\frac{m}{2}}}[i]$ and $\mathbb{Z}_{2^{\frac{m+1}{2}}} + i\mathbb{Z}_{2^{\frac{m-1}{2}}}$ when $m$ is even and odd, respectively. In order to use SBWD, and to decode to a codeword within the code, we \emph{trim} the in-phase and quadrature components of the received vector (the algorithm is given below) to be within a box $\mathcal{B}^{\prime} \supseteq \mathcal{B}$ marginally larger than $\mathcal{B}$ by length $\epsilon$ on each dimension. Then, we feed the trimmed received vector to the SBWD and decode the information bits. Note that the choice of $\epsilon$ is crucial to decode a codeword within the code, and to improve the BER with reference to the SBWD. We now provide an algorithm for the trimming method, which works independently on the in-phase and quadrature component of the scalars in $\textbf{y} = [y_{1}, y_{2}, \ldots, y_{2^{m}}]$ in \eqref{eq_awgn_channel}. In particular, the algorithm presented in the sequel works on the in-phase and quadrature component of $y_{j}$ when $m$ is even. Extension to the case when $m$ is odd is straightforward.

\begin{flushleft}
\textbf{Algorithm for the trimming technique when $m$ is even:}
\end{flushleft}

\begin{mdframed}
\begin{flushleft}
Input $y \in \mathbb{R}$ (either $\Re(y_{j})$ or $\Im(y_{j})$)\\
\end{flushleft}

\noindent \textbf{function} TRIM$(y$, $\epsilon)$\\
$~~~$ $\Delta$ = $(2^{\frac{m}{2}} - 1)/2$\\
$~~~$ $r$ = $y$ - $\Delta$\\
$~~~$ $t$ = $\Delta$ + $\epsilon$\\
$~~~$ \textbf{if} $|r| > t$\\
$~~~$ $~~~$ $s$ = $t/|r|$\\
$~~~$ $~~~$ $b = s \times r$\\
$~~~$ \textbf{else}\\
$~~~$ $~~~$ $b = r$\\
$~~~$ \textbf{end if}\\
$~~~$ return $b$ + $\Delta$\\
\textbf{end function}
\end{mdframed}
\vspace{.2cm}

Using BWCD, we have obtained BER for dimensions when $m = 2, 4, \mbox{ and } 6$, and compared them with the BER of the SBWD. The plots as shown in Fig. \ref{ber_ibwd_fbwd} indicate that BWCD outperforms SBWD by 0.5 dB. For the presented results, we have used $\epsilon = \frac{1}{2\sqrt{2}}$, which corresponds to the packing radius of $\frac{\sqrt{N}}{2}$. The above value of $\epsilon$ was optimized based on the simulation results by comparing the BER for various values of $\epsilon$.  Intuitively, trimming the received vector to fall within the packing radius of a lattice point in the edge of the lattice code forces to SBWD to decode to a lattice point in the edge of the code rather than a lattice point outside the lattice code.

\begin{figure}[h]
\centering
\includegraphics[width=3.5in]{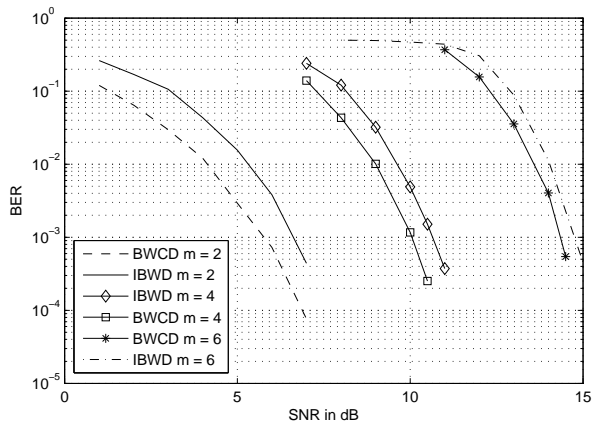}
\caption{Comparison of BER between BWCD and SBWD for $m = 2, 4, \mbox{ and } 6$.}
\label{ber_ibwd_fbwd}
\end{figure}

\section{Conclusion and Directions for Future Work}
\label{sec7}

In the first part of this paper, we have introduced a new method of encoding complex BW lattices, which facilitates bit labelling of BW lattice points. As a generalization, the proposed technique is applicable to encode all Construction $D$ complex lattices. In the second part of this paper, we have used complex BW lattice codes for communication over AWGN channels. To encode the code, we have used Construction $A^{\prime}$, and to decode the code we have used the SBWD. We have studied the error performance of the SBWD, and have shown that the Jacobi-Theta functions can characterize the virtual binary channels that arise in the decoding process. We have also shown that the SBWD is powerful in making correct decisions beyond the packing radius. Subsequently, we have used the SBWD to decode the complex lattice code through the noise trimming technique. This is the first work that uncovers the potential of SBWD (in terms of the error performance) in decoding lattice codes of large-block lengths with low-complexity. This work can be extended in one of the following ways:
\begin{itemize}
\item The SBWD proposed in \cite{MiN} uses a soft-input, hard-output RM decoder at each level of Construction $D$. It will be interesting to study the error performance of the lattice decoder with soft-input, soft-output iterative RM decoders.
\item We have presented the error performance of the SBWD through simulation results, and hence we now know the SBWD error performance with reference to the sphere lower bound and the sphere upper bound. A closed form expression on the error performance of the SBWD could be obtained for a better understanding of the decoder performance.
\end{itemize}
\section*{Acknowledgment}
This work was performed within the Monash Software Defined Telecommunications Lab and supported by the Monash Professional Fellowship 2012-2013 and DP 130100103.

\end{document}